\documentclass[11pt,reqno]{amsart}

\usepackage{amsmath,amssymb,amsfonts,amsthm,mathtools,mathrsfs}
\usepackage[margin=1.05in]{geometry}
\usepackage{enumitem}
\usepackage[hidelinks]{hyperref}

\newcommand{\II}{\mathrm{II}}
\newcommand{\IIo}{\mathring{\mathrm{II}}}

\numberwithin{equation}{section}

\theoremstyle{plain}
\newtheorem{theorem}{Theorem}[section]
\newtheorem{proposition}[theorem]{Proposition}
\newtheorem{lemma}[theorem]{Lemma}
\newtheorem{corollary}[theorem]{Corollary}

\theoremstyle{definition}
\newtheorem{definition}[theorem]{Definition}

\theoremstyle{remark}
\newtheorem{remark}[theorem]{Remark}

\title[Charged Hawking mass on CMC collars]
{Charged Hawking Mass on CMC Collars
\\with Charged Matter}

\author{Sehong Park}
\address{Human-Centered Artificial Intelligence Research Institute,
Ewha Womans University, Seoul, Republic of Korea}
\email{pshong21@ewha.ac.kr}

\hypersetup{
  pdftitle={Charged Hawking Mass on CMC Collars with Charged Matter},
  pdfauthor={Sehong Park},
  pdfsubject={Exact charged Hawking-mass identities and rigidity on positive-Jacobi CMC collars},
  pdfkeywords={charged Hawking mass, charged matter, Reissner-Nordstrom, CMC collar, rigidity, horizon Hessian}
}

\begin{document}

\begin{abstract}
\sloppy
We study strictly stable minimal horizons in three-dimensional
time-symmetric charged initial data sets with cosmological constant,
allowing charged matter. Along the constant-mean-curvature foliation
issuing from the horizon through leaves with positive Jacobi operator,
we prove exact first-variation identities for the charged Hawking mass. When the charge
density has one sign compatible with the horizon charge, Hawking masses are nondecreasing, and equality between the two boundary
values of the fixed-charge mass forces the entire collar to be isometric to an annulus in a source-free constant-curvature
Reissner--Nordstr\"om--(anti-)de Sitter model. From this we recover the local rigidity results under a local maximality condition.  At the minimal endpoint the defect entering these identities computes the fixed-charge
Hawking-mass Hessian along the normalized lapse, and its vanishing
characterizes the model horizon data.
\end{abstract}

\maketitle

\section{Introduction}
\label{sec:intro}

Let $(M^3,\bar g,E)$ be a three-dimensional time-symmetric charged
initial data set with cosmological constant $\Lambda\in\mathbb R$, so
that
\[
 R_{\bar g}\ge 2|E|^2+2\Lambda ,
\]
and let $\Sigma_0\subset M$ be a connected, two-sided, strictly stable
minimal surface.  Strict stability determines a locally unique
constant-mean-curvature branch $\{\Sigma_h\}$ issuing from $\Sigma_0$,
with $H(\Sigma_h)=h$, whose normal lapse solves
$\mathcal J_h\rho_h=1$ and is positive for $|h|$ small.  For the finite-region statements, we work on an embedded
positive-Jacobi CMC collar
$\{\Sigma_h\}_{0\le h\le h_*}$ and set
$\Sigma_*:=\Sigma_{h_*}$.

The flux
$q_h:=(4\pi)^{-1}\int_{\Sigma_h}\langle E,\nu_h\rangle\,dA_h$ through the
leaves obeys
\begin{equation}\label{eq:intro-flux}
 \frac{d}{dh}q_h=q_h'
 =
 \frac1{4\pi}\int_{\Sigma_h}
 \rho_h\,\operatorname{div}_{\bar g}E\,dA_h .
\end{equation}
Thus $\operatorname{div}_{\bar g}E=0$ implies conservation of the leaf
charge, whereas a nonzero charge source allows the flux to vary along
the CMC family.

For a closed orientable surface $\Gamma$ and a charge parameter
$a\in\mathbb R$ we use
\begin{equation}\label{eq:intro-hawking}
 m_{\Lambda,a}(\Gamma)
 =
 \sqrt{\frac{|\Gamma|}{16\pi}}
 \left(
  \frac{\chi(\Gamma)}2
  +\frac{4\pi a^2}{|\Gamma|}
  -\frac{\Lambda|\Gamma|}{12\pi}
  -\frac1{16\pi}\int_\Gamma H^2\,dA
 \right),
\end{equation}
which for $\Lambda=0$, spherical topology, and $a$ equal to the flux is
the charged Hawking mass of Disconzi--Khuri \cite{DisconziKhuri}; see
also Jang \cite{Jang}.  The
model geometries are the spatial
Reissner--Nordstr\"om--(anti-)de Sitter metrics
\begin{equation}\label{eq:intro-model}
 \bar g=\frac{dR^2}{f(R)}+R^2\sigma_\kappa,
 \qquad
 E=\frac q{R^2}\sqrt{f(R)}\,\partial_R,
 \qquad
 f=\kappa-\frac{2m}R+\frac{q^2}{R^2}-\frac{\Lambda R^2}3,
\end{equation}
where $\sigma_\kappa$ is a constant-curvature metric of curvature
$\kappa$ and area $4\pi$ on a closed orientable surface, so that
$\kappa=\chi/2$ by Gauss--Bonnet. These models saturate the scalar constraint and have constant electric flux \(q\).  The normalization $\chi(\Gamma)/2$ in
\eqref{eq:intro-hawking} is chosen so that $m_{\Lambda,q}$ equals $m$ on
every coordinate surface of \eqref{eq:intro-model}, in every genus.
Along the CMC branch there are then two natural choices of charge
parameter: the fixed-charge mass $m_{\Lambda,q}(\Sigma_h)$ with
$q:=q_0$, and the current-charge mass $m_{\Lambda,q_h}(\Sigma_h)$.  They
agree when the flux is conserved.

Our first result is a pair of exact identities.  Write
$|\Sigma_h|=4\pi r_h^2$, let $\bar\rho_h$ be the average of $\rho_h$,
$w_h=\rho_h-\bar\rho_h$, and
\[
 \mathscr B_h
 =
 \mathscr D(\Sigma_h)
 +
 \frac1{\bar\rho_h^2}\mathcal Q_h(w_h,w_h),
\]
where $\mathscr D$ is the nonnegative geometric and electromagnetic
defect defined in \eqref{eq:defect}.  Then
\begin{align}
 \frac d{dh}m_{\Lambda,q}(\Sigma_h)
 &=
 \frac{r_hh\bar\rho_h}{16\pi}\mathscr B_h
 +
 \frac{h\bar\rho_h}{4r_h}(q_h^2-q^2),
 \label{eq:intro-fixed}\\
 \frac d{dh}m_{\Lambda,q_h}(\Sigma_h)
 &=
 \frac{r_hh\bar\rho_h}{16\pi}\mathscr B_h
 +
 \frac{q_hq_h'}{r_h}.
 \label{eq:intro-current}
\end{align}
The two
identities share $\mathscr B_h$ and differ only in how charge production
is recorded.

Suppose now that $\varsigma q\ge0$ and
$\varsigma\operatorname{div}_{\bar g}E\ge0$ for some
$\varsigma\in\{-1,1\}$.  Since $\rho_h>0$, \eqref{eq:intro-flux} gives
$q_h^2\ge q^2$ and $q_hq_h'\ge0$, so both masses are nondecreasing; if
$q=0$, the fixed-charge mass is nondecreasing with no condition on the
source.  Under either hypothesis,
\begin{equation}\label{eq:intro-chain}
 m_{\Lambda,q_*}(\Sigma_*)
 \ge
 m_{\Lambda,q}(\Sigma_*)
 \ge
 m_{\Lambda,q}(\Sigma_0)=:m_0 .
\end{equation}
In the same-sign case the first equality holds precisely when
$\operatorname{div}_{\bar g}E$ vanishes on the collar, and carries no
geometric information.

The second equality is rigidity: it holds if and only if the collar is a
source-free model region \eqref{eq:intro-model}, with
$\kappa=\chi(\Sigma_0)/2$, mass $m_0$, and $\sigma_\kappa$ inherited
from the horizon. No sign of $\Lambda$ is assumed a priori.  Strict stability instead
imposes the horizon compatibility condition
\[
 \kappa-\frac{q^2}{r^2}-\Lambda r^2>0;
\]
in particular, $\kappa\le0$ forces $\Lambda<0$.  In the spherical case
$\kappa=1$, once equality identifies the collar with the model, the
positive-Jacobi condition holds precisely for
\[
 r<R<R_{\mathrm{ph}}^+(m_0,q)
 :=
 \frac12\left(3m_0+\sqrt{9m_0^2-8q^2}\right),
\]
and degenerates at $R=R_{\mathrm{ph}}^+(m_0,q)$.

On a sufficiently short one-sided branch no continuation hypothesis is
needed.  If $\Sigma_0$ locally maximizes the fixed-charge mass, this
bounds the mass above along the branch, \eqref{eq:intro-fixed} bounds
it below, and equality forces the model geometry. A two-sided neighborhood
follows when $q=0$, or when the same-sign condition holds on each side of
$\Sigma_0$ for the normal pointing into that side.

At the strictly stable minimal endpoint, the same defect appears in the
second variation of the fixed-charge Hawking mass. Set
$|\Sigma_0|=4\pi r^2$, keep $q$ fixed, and let $\rho>0$ solve
$\mathcal J\rho=1$, with average $\bar\rho$, fluctuation $w$, and
$p=\rho/\sqrt{\bar\rho}$.  Then
\begin{equation}\label{eq:intro-endpoint}
 \frac{16\pi}r\left.m_{\Lambda,q}''\right|_p
 =
 \mathscr D(\Sigma_0)+\frac1{\bar\rho^2}\mathcal Q(w,w)
 =
 \mathscr B_0\ge0,
 \qquad
 \left.\frac{d^2}{dh^2}\right|_{h=0}m_{\Lambda,q}(\Sigma_h)
 =\bar\rho\left.m_{\Lambda,q}''\right|_p .
\end{equation}
Vanishing is equivalent to the model horizon data and to nonpositivity
of the full fixed-charge Hessian.  In particular, a nonmodel strictly
stable horizon is never a local maximizer: its own CMC branch selects a
direction of strictly positive response.
Appendix~\ref{app:concentrated-shell} constructs collars whose horizons
carry exact model data and for which
$m_{\Lambda,q}(\Sigma_*)-m_0\to0$, yet which do not converge to the
model in $C^2$. Thus model horizon data alone do not determine the ambient collar;
rigidity of the whole collar follows instead from the fixed-charge
boundary equality
\[
m_{\Lambda,q}(\Sigma_*)=m_{\Lambda,q}(\Sigma_0).
\]

The local rigidity consequence recovers the Hawking-mass rigidity results
of M\'aximo and Nunes \cite{MaximoNunes} and Barros, Batista and Cruz
\cite{BarrosBatistaCruz}, together with their charged analogues due to
Baltazar, Barros and Batista \cite{BaltazarBarrosBatista} and Sousa and
Lima \cite{SousaLima}.  These works treat the positive- and
negative-cosmological-constant regimes separately and, in the charged
cases, assume conserved electric flux, whereas here $\Lambda\in\mathbb R$
is arbitrary and the leaf charge may vary.  The same-sign condition used
above is the one appearing in McCormick's charged Penrose inequality
\cite{McCormickMatter}; see also Khuri, Weinstein and Yamada
\cite{KWY,KWYExtensions} for the global charged Penrose problem.  The
finite-boundary comparison is related to the localized Riemannian Penrose
inequality of Miao \cite{Miao2009}, the CMC boundary estimates of
Miao--Wang--Xie \cite{MiaoWangXie}, and the charged quasi-local Penrose
inequalities of Chen--McCormick \cite{ChenMcCormick}.

Section~\ref{sec:setup} develops the CMC branch, the finite
positive-Jacobi continuation, the flux law, and the two masses.
Section~\ref{sec:exact-defect} proves
\eqref{eq:intro-fixed}--\eqref{eq:intro-current} and their monotonicity
consequences.  Section~\ref{sec:rigidity} treats the two equality cases
and the Reissner--Nordstr\"om--(anti-)de Sitter rigidity theorem.
Section~\ref{sec:endpoint-hessian} proves \eqref{eq:intro-endpoint} and
the horizon characterizations.  Appendix~\ref{app:matter-example} gives
sourced CMC collars and the concentrated-shell example.

\section{The CMC branch, flux, and the defect}
\label{sec:setup}

\subsection{Data and the CMC branch}
\label{subsec:data}

Throughout, $(M^3,\bar g,E)$ is a smooth three-manifold with a smooth
vector field $E$, $\Lambda\in\mathbb R$, and
$\zeta:=\operatorname{div}_{\bar g}E$; the scalar constraint
$R_{\bar g}\ge2|E|^2+2\Lambda$ is assumed on every region considered.
Surfaces $\Gamma\subset M$ are closed, connected, orientable, embedded,
and two-sided, with unit normal $\nu$, Gaussian curvature $K$, mean
curvature $H$ with respect to $\nu$, second fundamental form $\II$ with
trace-free part $\IIo$, area radius $|\Gamma|=4\pi r_\Gamma^2$,
components $E=E^\top+E_\nu\,\nu$ along $\Gamma$, and flux
$q_\Gamma=(4\pi)^{-1}\int_\Gamma E_\nu\,dA$.  The Jacobi operator is
\begin{equation}\label{eq:jacobi}
 \mathcal J=-\Delta_\Gamma-S,
 \qquad
 S=\operatorname{Ric}(\nu,\nu)+|\II|^2,
 \qquad
 \mathcal Q(u,v)=\int_\Gamma u\,\mathcal Jv\,dA,
\end{equation}
with first eigenvalue $\lambda_1(\mathcal J)$; a minimal surface is
strictly stable when $\lambda_1(\mathcal J)>0$.  Our variation
convention is $H'=\mathcal J\rho$ for normal speed $\rho$.  The
same-sign condition on a region is
\begin{equation}\label{eq:same-sign-intro}
 \varsigma q\ge0
 \quad\text{and}\quad
 \varsigma\zeta\ge0
 \qquad\text{for some }\varsigma\in\{-1,1\},
\end{equation}
where $q$ is the horizon flux fixed below.

\begin{definition}[Charged positive-Jacobi CMC region]
\label{def:fillin}
A charged positive-Jacobi CMC region (a \emph{collar}) is a compact region
$\Omega\subset M$ swept out by a smooth family of embeddings
$F_h:\Sigma_0\to M$, with $\Sigma_h=F_h(\Sigma_0)$ and
$h\in[0,h_*]$, such that $(h,x)\mapsto F_h(x)$ is a diffeomorphism
onto $\Omega$ and:
\begin{enumerate}[label=\textup{(\roman*)},itemsep=1pt]
\item $H(\Sigma_h)=h$ and $\Sigma_0$ is a strictly stable minimal surface;
\item $\lambda_1(\mathcal J_h)>0$ and
\[
\rho_h:=\langle\partial_hF_h,\nu_h\rangle>0
\]
for every $h\in[0,h_*]$;
\item $\partial\Omega=\Sigma_0\sqcup\Sigma_*$, where
$\Sigma_*=\Sigma_{h_*}$.
\end{enumerate}
The unit normals $\nu_h$ are chosen to point in the direction of increasing $h$.
We write $r_h=r_{\Sigma_h}$, $q_h=q_{\Sigma_h}$, $q=q_0$,
$\bar\rho_h$ for the average of $\rho_h$, and
$w_h=\rho_h-\bar\rho_h$.
\end{definition}

\begin{remark}\label{rem:collar-function}
The map $\tau:\Omega\to[0,h_*]$ with $\tau\equiv h$ on $\Sigma_h$ is
smooth, and since $\tau$ is leafwise constant its gradient is normal to
every leaf: $\nabla\tau=|\nabla\tau|\nu_h$ with
$|\nabla\tau|=\rho_h^{-1}$.  By the coarea formula the volume element is
$dV_{\bar g}=\rho_h\,dA_h\,dh$.
\end{remark}

\begin{lemma}[Canonical CMC branch]
\label{lem:branch}
Let $\Sigma\subset M$ be a closed, two-sided, strictly stable minimal
surface.  Then there exist $\varepsilon>0$ and a unique smooth
normal-graph family $\{\Sigma_h\}_{|h|<\varepsilon}$ with
$\Sigma_h|_{h=0}=\Sigma$ and $H(\Sigma_h)=h$; its lapse solves
$\mathcal J_h\rho_h=1$, and after shrinking $\varepsilon$ both
$\rho_h>0$ and $\lambda_1(\mathcal J_h)>0$ hold on the family.
\end{lemma}

\begin{proof}
Let $\mathcal H(u)$ denote the mean curvature of the normal graph of $u$ over $\Sigma$.
Since $D\mathcal H(0)=\mathcal J$ is invertible, the implicit
function theorem gives a unique local family $u(h)$ satisfying
$\mathcal H(u(h))=h$.
Differentiating yields $\mathcal J_h\rho_h=1$, and the maximum
principle associated with $\lambda_1(\mathcal J)>0$ gives
$\rho_0>0$.
Since $\partial_hu(0,\cdot)=\rho_0$, after shrinking $\varepsilon$
we have $\partial_hu(h,\cdot)>0$, so the normal graphs foliate a
neighborhood of $\Sigma$.
By continuity, after shrinking $\varepsilon$ again if necessary,
$\rho_h>0$ and $\lambda_1(\mathcal J_h)>0$ throughout the family.
\end{proof}

\subsection{Flux and charge along the leaves}
\label{subsec:flux}

\begin{lemma}[Flux laws]
\label{lem:collar-basics}
On a charged positive-Jacobi CMC region,
\begin{align}
 \frac{d}{dh}q_h=q_h'&=\frac1{4\pi}\int_{\Sigma_h}\rho_h\,\zeta\,dA_h,
 \label{eq:flux-evolution}\\
 q_{h_2}-q_{h_1}
 &=\frac1{4\pi}\int_{\Omega_{h_1,h_2}}\zeta\,dV_{\bar g},
 \qquad 0\le h_1<h_2\le h_*,
 \label{eq:flux-slab}
\end{align}
where $\Omega_{h_1,h_2}$ is the slab between the two leaves, and on
every leaf
\begin{equation}\label{eq:leaf-charge-square}
 \int_{\Sigma_h}|E|^2\,dA_h
 =\int_{\Sigma_h}|E_h^\top|^2\,dA_h
 +\int_{\Sigma_h}\Bigl(E_{\nu_h}-\frac{q_h}{r_h^2}\Bigr)^2dA_h
 +\frac{4\pi q_h^2}{r_h^2}.
\end{equation}
Under \eqref{eq:same-sign-intro}, $\varsigma q_h$ is nondecreasing, so
\begin{equation}\label{eq:charge-monotonicity}
 \varsigma q_h\ge\varsigma q\ge0,
 \qquad
 q_h^2\ge q^2
 \qquad\text{for all }h.
\end{equation}
\end{lemma}

\begin{proof}
\eqref{eq:flux-slab} is the divergence theorem with the outward
orientations matched along the branch, and differentiating it with
$dV_{\bar g}=\rho_h\,dA_h\,dh$ from
Remark~\ref{rem:collar-function} gives \eqref{eq:flux-evolution}.  For
\eqref{eq:leaf-charge-square}, decompose
$E=E_h^\top+E_{\nu_h}\nu_h$ and expand the square: since
$\int E_{\nu_h}=4\pi q_h$ and $|\Sigma_h|=4\pi r_h^2$,
\[
 \int_{\Sigma_h}\Bigl(E_{\nu_h}-\frac{q_h}{r_h^2}\Bigr)^2
 =\int_{\Sigma_h}E_{\nu_h}^2-\frac{4\pi q_h^2}{r_h^2}.
\]
Finally, $\rho_h>0$ and \eqref{eq:same-sign-intro} give
$\varsigma q_h'\ge0$ in \eqref{eq:flux-evolution}, hence
\eqref{eq:charge-monotonicity}.
\end{proof}

\subsection{Hawking masses and the defect}
\label{subsec:masses}

On a CMC leaf, $H\equiv h$ and $\chi(\Sigma_h)=\chi(\Sigma_0)=2\kappa$,
so \eqref{eq:intro-hawking} reduces to
\begin{equation}\label{eq:hawking-cmc}
 m_{\Lambda,a}(\Sigma_h)
 =\frac{r_h}2
 \Bigl(\kappa+\frac{a^2}{r_h^2}-\frac{\Lambda r_h^2}3
 -\frac{r_h^2h^2}4\Bigr),
\end{equation}
and the two charge choices differ by an explicit gap:
\begin{equation}\label{eq:current-fixed-mass}
 m_{\Lambda,q_h}(\Sigma_h)
 =m_{\Lambda,q}(\Sigma_h)+\frac{q_h^2-q^2}{2r_h}.
\end{equation}
With $m_0:=m_{\Lambda,q}(\Sigma_0)$ we also set
\begin{equation}\label{eq:f0}
 f_0(R):=\kappa-\frac{2m_0}{R}+\frac{q^2}{R^2}-\frac{\Lambda R^2}{3},
\end{equation}
the function $f$ of \eqref{eq:intro-model} with mass $m_0$ and charge
$q$.

For a leaf $\Gamma$ with flux $q_\Gamma$, define the defect
\begin{equation}\label{eq:defect}
 \mathscr D(\Gamma)
 =\frac12\int_\Gamma\bigl(R_{\bar g}-2|E|^2-2\Lambda\bigr)dA
 +\frac12\int_\Gamma|\IIo|^2\,dA
 +\int_\Gamma|E^\top|^2\,dA
 +\int_\Gamma\Bigl(E_\nu-\frac{q_\Gamma}{r_\Gamma^2}\Bigr)^2dA .
\end{equation}
Each term is nonnegative under the scalar constraint. We also define
\begin{equation}\label{eq:production-defect}
 \mathscr B_h
 =\mathscr D(\Sigma_h)
 +\frac1{\bar\rho_h^2}\mathcal Q_h(w_h,w_h)\ge0 .
\end{equation}

The defect vanishes exactly when
$R_{\bar g}-2|E|^2-2\Lambda=0$, the leaf
is umbilic, and the electric field is normal with leafwise constant
flux density; the lapse term vanishes exactly when $w_h=0$.

\subsection{The horizon stability bound}
\label{subsec:constant-slack}

For a strictly stable minimal $\Sigma_0$ with $|\Sigma_0|=4\pi r^2$ and
flux $q$, set
\[
 \Theta_{\kappa,\Lambda}
 =\kappa-\frac{q^2}{r^2}-\Lambda r^2 .
\]

\begin{lemma}[Horizon stability bound]
\label{lem:constant-slack}
Under the scalar constraint,
\begin{equation}\label{eq:constant-slack}
 \mathcal Q(1,1)=4\pi\Theta_{\kappa,\Lambda}-\mathscr D(\Sigma_0),
\end{equation}
and hence
\begin{equation}\label{eq:auto-subextremal}
 \lambda_1(\mathcal J)\,r^2
 \ \le\
 \Theta_{\kappa,\Lambda}-\frac{\mathscr D(\Sigma_0)}{4\pi}
 \ \le\
 \Theta_{\kappa,\Lambda},
\end{equation}
so strict stability forces $\Theta_{\kappa,\Lambda}>0$.  Consequently
\begin{equation}\label{eq:horizon-eigenvalue-bound}
 \bigl(\Lambda+\lambda_1(\mathcal J)\bigr)r^4-\kappa r^2+q^2\le0,
 \qquad
 \Lambda r^4-\kappa r^2+q^2<0 ;
\end{equation}
in particular $\kappa\le0$ forces $\Lambda<0$, and if
$\Lambda+\lambda_1(\mathcal J)>0$, then $\kappa=1$ and
$4(\Lambda+\lambda_1(\mathcal J))q^2\le1$.
\end{lemma}

\begin{proof}
Since $H=0$, the Gauss equation gives
$S=R/2-K+|\IIo|^2/2$, and hence
\[
 \mathcal Q(1,1)
 =4\pi\kappa-\frac12\int_{\Sigma_0}(R+|\IIo|^2)\,dA .
\]
Together with \eqref{eq:defect} and \eqref{eq:leaf-charge-square}, this
gives \eqref{eq:constant-slack}.  By the variational characterization
of $\lambda_1(\mathcal J)$,
$\lambda_1|\Sigma_0|\le\mathcal Q(1,1)$, which yields
\eqref{eq:auto-subextremal} and
$\Theta_{\kappa,\Lambda}\ge\lambda_1r^2>0$.

Multiplying \eqref{eq:auto-subextremal} by $r^2$ gives the first
inequality in \eqref{eq:horizon-eigenvalue-bound}; since
$\lambda_1>0$, the second follows.  Thus $\kappa\le0$ forces
$\Lambda<0$.  Finally, if $\Lambda+\lambda_1>0$, the quadratic in
$r^2$ can be nonpositive for some $r^2>0$ only if $\kappa=1$ and
$1-4(\Lambda+\lambda_1)q^2\ge0$.
\end{proof}

Equation~\eqref{eq:horizon-eigenvalue-bound} refines area--charge
bounds with cosmological constant; compare
\cite{SimonAreaCharge,HaywardShiromizuNakao,Woolgar,DainJaramilloReiris};
for rigidity in the equality cases of such bounds, see
\cite{MendesAreaCharge,GallowayMendesCharged,LimaSousaBatista}.
No sign of $m_0$ is assumed when $\kappa<0$, while $\kappa=0$ with
$\Lambda<0$ gives $m_0>0$ from \eqref{eq:hawking-cmc}.

\section{Hawking-mass variation along the CMC continuation}
\label{sec:exact-defect}

We now compute the variation of the charged Hawking mass along the CMC continuation.  The calculation yields exact identities for both the fixed-charge and current-charge masses in terms of the defect \(\mathscr B_h\) and the variation of the leaf charge.

\subsection{CMC and Jacobi variation algebra}
\label{subsec:variation-algebra}

Since the normal speed is $\rho_h$ and
$H_h=h$, the first variation of area gives
\[
 \frac d{dh}|\Sigma_h|
 =\int_{\Sigma_h}H_h\rho_h\,dA_h
 =4\pi r_h^2h\bar\rho_h.
\]
Thus
\begin{equation}\label{eq:r-derivative}
 r_h'=\frac12r_hh\bar\rho_h.
\end{equation}
Expanding the CMC formula \eqref{eq:hawking-cmc} along the continuation gives
the equivalent form
\begin{equation}\label{eq:cmc-mass-specialization}
 m_{\Lambda,a}(\Sigma_h)
 =\frac{\kappa r_h}{2}+\frac{a^2}{2r_h}
  -\frac{\Lambda r_h^3}{6}-\frac{r_h^3h^2}{8}.
\end{equation}

The lapse equation \(\mathcal J_h\rho_h=1\), with
$S_h=\operatorname{Ric}(\nu_h,\nu_h)+|\II_h|^2$ as in
\eqref{eq:jacobi}, has two useful integrated consequences.  First, since
\(\rho_h=\bar\rho_h+w_h\) and
\(\int_{\Sigma_h}w_h\,dA_h=0\), integration gives
\begin{equation}\label{eq:integrated-lapse}
\bar\rho_h\int_{\Sigma_h}S_h\,dA_h
+\int_{\Sigma_h}S_hw_h\,dA_h
=-4\pi r_h^2.
\end{equation}

Second, self-adjointness and the zero mean of $w_h$ give
\[
 0=\mathcal Q_h(\rho_h,w_h)
  =\bar\rho_h\mathcal Q_h(1,w_h)+\mathcal Q_h(w_h,w_h).
\]
Because
$\mathcal Q_h(1,w_h)=-\int_{\Sigma_h}S_hw_h\,dA_h$, we obtain the exact
fluctuation identity
\begin{equation}\label{eq:lapse-identity}
 \int_{\Sigma_h}S_hw_h\,dA_h
 =\frac1{\bar\rho_h}\mathcal Q_h(w_h,w_h)\ge0.
\end{equation}
The final inequality follows from the positive definiteness of $\mathcal Q_h$.

The Gauss equation on the CMC leaf reads
\[
 S_h
 =\frac12\bigl(R_{\bar g}+|\IIo_h|^2\bigr)
  -K_h+\frac34h^2.
\]
Integrating and using Gauss--Bonnet gives
\begin{equation}\label{eq:S-average}
 \frac1{4\pi}\int_{\Sigma_h}S_h\,dA_h
 =\frac1{8\pi}\int_{\Sigma_h}
   \bigl(R_{\bar g}+|\IIo_h|^2\bigr)\,dA_h
  -\kappa+\frac34r_h^2h^2.
\end{equation}

Finally, since
\[
 \frac{1}{|\Sigma_h|}
 \int_{\Sigma_h} E_{\nu_h}\,dA_h
 =\frac{q_h}{r_h^2},
\]
we obtain
\[
 \frac{4\pi q_h^2}{r_h^2}
 =\int_{\Sigma_h}|E|^2\,dA_h
  -\int_{\Sigma_h}|E_h^\top|^2\,dA_h
  -\int_{\Sigma_h}
   \left(E_{\nu_h}-\frac{q_h}{r_h^2}\right)^2dA_h.
\]

\subsection{Fixed- and current-charge variation formulas}
\label{subsec:exact-production-laws}

\begin{theorem}[Exact fixed- and current-charge production laws]
\label{thm:radial-identity}
Let $(\Omega,\bar g,E)$ be a charged positive-Jacobi CMC region.  Then
the fixed-charge Hawking mass satisfies
\begin{equation}\label{eq:exact-radial-fixed}
\frac d{dh}m_{\Lambda,q}(\Sigma_h)
=\frac{r_hh\bar\rho_h}{16\pi}\mathscr B_h
+\frac{h\bar\rho_h}{4r_h}(q_h^2-q^2),
\end{equation}
where $q=q_0$, while the current-charge Hawking mass satisfies
\begin{equation}\label{eq:exact-radial-current}
 \frac d{dh}m_{\Lambda,q_h}(\Sigma_h)
 =
 \frac{r_hh\bar\rho_h}{16\pi}\mathscr B_h
 +\frac{q_hq_h'}{r_h}.
\end{equation}
\end{theorem}

\begin{proof}
Set $a=q$ in \eqref{eq:cmc-mass-specialization}.  Differentiating and
using \eqref{eq:r-derivative} gives
\[
 \frac d{dh}m_{\Lambda,q}(\Sigma_h)
 =\frac{r_hh\bar\rho_h}{4}
  \left(\kappa-\frac{q^2}{r_h^2}-\Lambda r_h^2
        -\frac34r_h^2h^2\right)-\frac14r_h^3h.
\]
Using \eqref{eq:integrated-lapse} and \eqref{eq:S-average}, we obtain
\begin{equation}\label{eq:radial-intermediate}
 \begin{aligned}
  \frac d{dh}m_{\Lambda,q}(\Sigma_h)
  ={}&\frac{r_hh\bar\rho_h}{4}
  \left[
   \frac1{8\pi}\int_{\Sigma_h}
   \bigl(R_{\bar g}+|\IIo_h|^2\bigr)\,dA_h
   -\frac{q^2}{r_h^2}-\Lambda r_h^2
  \right]\\
  &+\frac{r_hh}{16\pi}
    \int_{\Sigma_h}S_hw_h\,dA_h .
 \end{aligned}
\end{equation}

Using \eqref{eq:leaf-charge-square}, the bracket in
\eqref{eq:radial-intermediate} is exactly
\[
 \frac1{4\pi}\mathscr D(\Sigma_h)
 +\frac{q_h^2-q^2}{r_h^2}.
\]
Together with \eqref{eq:lapse-identity} and
\eqref{eq:production-defect}, this gives
\eqref{eq:exact-radial-fixed}.

For the current-charge mass, differentiate the algebraic relation
\eqref{eq:current-fixed-mass}.  Equation~\eqref{eq:r-derivative} gives
\[
 \frac d{dh}\left(\frac{q_h^2-q^2}{2r_h}\right)
 =\frac{q_hq_h'}{r_h}
  -\frac{h\bar\rho_h}{4r_h}(q_h^2-q^2).
\]
The second term cancels the charge-accumulation term in
\eqref{eq:exact-radial-fixed}, which proves \eqref{eq:exact-radial-current}.
\end{proof}

First variation formulas for the charged Hawking mass under general normal variations appear in \cite{BaltazarBarrosBatista,SousaLima}.  Along the CMC family considered here, the choice \(H_h=h\) leads instead to the exact variation formulas above.  Their geometric part is collected in the nonnegative quantity \(\mathscr B_h\), which records the scalar-constraint, extrinsic-geometric, electromagnetic, and lapse defects relevant to rigidity, while the remaining term records the variation of the leaf charge.  In particular, conservation of the leaf charge is not required.

Related monotonicity results for generalized Hawking mass along charged collar extensions were obtained in
\cite{CabreraCederbaumGehringPenuela}, using a different foliation
construction with conserved electric flux.  Extension and fill-in
constructions with charge, or with scalar curvature lower bounds, appear
in \cite{AlaeeCabreraCederbaum,McCormickFillins}.

\subsection{Integrated balance laws and same-sign monotonicity}
\label{subsec:same-sign-monotonicity}

Before imposing a sign condition, integrate
\eqref{eq:exact-radial-fixed}--\eqref{eq:exact-radial-current} over
$0\le h_1\le h_2\le h_*$.  The fixed-charge balance is
\begin{equation}\label{eq:integrated-budget}
 \begin{aligned}
 &m_{\Lambda,q}(\Sigma_{h_2})-m_{\Lambda,q}(\Sigma_{h_1})\\
 ={}&\int_{h_1}^{h_2}\frac{r_hh\bar\rho_h}{16\pi}
       \mathscr B_h\,dh
 +\int_{h_1}^{h_2}
       \frac{h\bar\rho_h}{4r_h}(q_h^2-q^2)\,dh,
 \end{aligned}
\end{equation}
whereas the current-charge balance is
\begin{equation}\label{eq:integrated-current-budget}
 \begin{aligned}
 &m_{\Lambda,q_{h_2}}(\Sigma_{h_2})
   -m_{\Lambda,q_{h_1}}(\Sigma_{h_1})\\
 ={}&\int_{h_1}^{h_2}\frac{r_hh\bar\rho_h}{16\pi}
       \mathscr B_h\,dh
 +\int_{h_1}^{h_2}\frac{q_h}{4\pi r_h}
   \left(\int_{\Sigma_h}\rho_h\zeta\,dA_h\right)dh.
 \end{aligned}
\end{equation}

\begin{corollary}[Same-sign and neutral monotonicity]
\label{cor:same-sign-monotonicity}
If \eqref{eq:same-sign-intro} holds, then both
\[
 h\longmapsto m_{\Lambda,q}(\Sigma_h),
 \qquad
 h\longmapsto m_{\Lambda,q_h}(\Sigma_h)
\]
are nondecreasing.  Moreover, if $q=0$, then, without any sign condition
on the charge source,
\[
 h\longmapsto m_{\Lambda,0}(\Sigma_h)
\]
is nondecreasing.

If either $q=0$ or \eqref{eq:same-sign-intro} holds, then
\begin{equation}
\label{eq:sharp-mass-chain}
 m_{\Lambda,q_*}(\Sigma_*)
 \ge m_{\Lambda,q}(\Sigma_*)
 \ge m_{\Lambda,q}(\Sigma_0)=m_0.
\end{equation}
\par\smallskip\noindent
The second comparison is equivalent to
\begin{equation}\label{eq:sharp-ineq}
 \frac{r_*^2h_*^2}{4}\le f_0(r_*),
\end{equation}
with $f_0$ as in \eqref{eq:f0}, and hence
\begin{equation}\label{eq:sharp-H}
 h_*\le\frac{2\sqrt{f_0(r_*)}}{r_*}.
\end{equation}
\end{corollary}

\begin{proof}
The common defect $\mathscr B_h$ is nonnegative.  Under the same-sign
condition, \eqref{eq:charge-monotonicity} gives
$q_h^2-q^2\ge0$ and $q_hq_h'\ge0$, so the two production laws give both
monotonicity statements.  If $q=0$, the fixed-charge production law becomes
\[
 \frac d{dh}m_{\Lambda,0}(\Sigma_h)
 =\frac{r_hh\bar\rho_h}{16\pi}\mathscr B_h
  +\frac{h\bar\rho_h}{4r_h}q_h^2
 \ge0 .
\]

In either case, fixed-charge monotonicity gives the second comparison in
\eqref{eq:sharp-mass-chain}, while the first is the algebraic calibration gap
\[
 m_{\Lambda,q_*}(\Sigma_*)-m_{\Lambda,q}(\Sigma_*)
 =\frac{q_*^2-q^2}{2r_*}\ge0.
\]
For $q=0$, this gap is $q_*^2/(2r_*)\ge0$ without a sign condition.
Finally, \eqref{eq:cmc-mass-specialization} and \eqref{eq:f0}
give
\[
 m_{\Lambda,q}(\Sigma_*)-m_0
 =\frac{r_*}{2}
 \left(f_0(r_*)-\frac{r_*^2h_*^2}{4}\right),
\]
which proves \eqref{eq:sharp-ineq}--\eqref{eq:sharp-H}.
\end{proof}

For \(q\ne0\), the same-sign condition \eqref{eq:same-sign-intro}
agrees with the condition used by McCormick in the charged Penrose
inequality with charged matter \cite[Thm.~3.1]{McCormickMatter}, where
the mass is computed with the initial charge.

\section{Equality and RN--(A)dS rigidity}
\label{sec:rigidity}

Throughout, $(\Omega,\bar g,E)$ is a charged positive-Jacobi CMC region,
and we assume either $q=0$ or the same-sign condition
\eqref{eq:same-sign-intro}, so that the comparison chain
\eqref{eq:sharp-mass-chain} holds. We determine the equality cases in both inequalities.

\begin{theorem}[Reissner--Nordstr\"om--(anti-)de Sitter rigidity]
\label{thm:main}
Let $(\Omega,\bar g,E)$ be a charged positive-Jacobi CMC region, and
assume that either $q=0$ or the same-sign condition
\eqref{eq:same-sign-intro} holds.  Set
\[
m_0:=m_{\Lambda,q}(\Sigma_0),\qquad
r:=r_0,\qquad r_*:=r_{h_*},
\]
and let $\gamma_0$ be the metric induced on $\Sigma_0$.

\begin{enumerate}[label=\textup{(\roman*)},itemsep=2pt]
\item If \eqref{eq:same-sign-intro} holds, then
\[
m_{\Lambda,q_*}(\Sigma_*)
=
m_{\Lambda,q}(\Sigma_*)
\]
if and only if $\zeta\equiv0$ throughout $\Omega$.

\item The fixed-charge equality
\[
m_{\Lambda,q}(\Sigma_*)
=
m_{\Lambda,q}(\Sigma_0)
\]
holds if and only if $(\Omega,\bar g,E)$ is isometric to the compact
collar bounded by the horizon $R=r$ and the coordinate surface
$R=r_*$ in the smooth spatial
Reissner--Nordstr\"om--(anti-)de Sitter model
\eqref{eq:intro-model}, with mass $m_0$, charge $q$, and
\[
\sigma_\kappa=r^{-2}\gamma_0,
\qquad
\kappa=\frac{\chi(\Sigma_0)}2.
\]
\end{enumerate}
\end{theorem}

The two equalities encode different information.  The first concerns
only the charge gap: under \eqref{eq:same-sign-intro}, it is equivalent
to $\zeta\equiv0$ on $\Omega$, but it imposes no condition on the
defects collected in $\mathscr B_h$.  By contrast, the fixed-charge
boundary equality forces the entire collar to be the corresponding
model collar.  If $q=0$ and no sign condition is imposed, the first
equality yields only $q_*=0$ and does not imply $\zeta\equiv0$.

\subsection{Vanishing of production}
\label{subsec:vanishing-production}

\begin{proposition}[Vanishing of production]
\label{prop:vanishing-production}
Suppose $m_{\Lambda,q}(\Sigma_*)=m_{\Lambda,q}(\Sigma_0)$.  Then, on
every leaf,
\begin{equation}\label{eq:defect-zero}
 \mathscr D(\Sigma_h)=0,\qquad
 \mathcal Q_h(w_h,w_h)=0,\qquad
 q_h=q,
\end{equation}
and hence pointwise
\begin{equation}\label{eq:model-equalities}
 R_{\bar g}=2|E|^2+2\Lambda,\qquad
 \IIo_h=0,\qquad
 E_h^\top=0,\qquad
 E_{\nu_h}=\frac q{r_h^2},
\end{equation}
together with $w_h=0$.  Moreover $\zeta=0$ throughout $\Omega$; if
$q=0$, then in addition $E\equiv0$.
\end{proposition}

\begin{proof}
By Corollary~\ref{cor:same-sign-monotonicity}, the fixed-charge mass is
nondecreasing along the branch.  Since its endpoint values agree, it is
constant.  Hence, for every $h\in(0,h_*]$,
\begin{align*}
0
=\frac d{dh}m_{\Lambda,q}(\Sigma_h)
=\frac{r_hh\bar\rho_h}{16\pi}\mathscr D(\Sigma_h)
 +\frac{r_hh}{16\pi\bar\rho_h}
   \mathcal Q_h(w_h,w_h)
 +\frac{h\bar\rho_h}{4r_h}(q_h^2-q^2).
\end{align*}
By the standing hypotheses,
$q_h^2-q^2\ge0$, while
$\mathscr D(\Sigma_h)\ge0$ and
$\mathcal Q_h(w_h,w_h)\ge0$.  Since all three coefficients are
positive for $h>0$, we obtain
\[
\mathscr D(\Sigma_h)=0,\qquad
\mathcal Q_h(w_h,w_h)=0,\qquad
q_h^2=q^2.
\]
The positive definiteness of $\mathcal Q_h$ gives $w_h=0$.  If $q=0$,
then $q_h=0$; if $q\ne0$, the same-sign condition together with
\eqref{eq:charge-monotonicity} gives $q_h=q$.  By continuity, these
conclusions also hold at $h=0$.

Returning to the definition of the defect, we have on every leaf
\begin{align*}
0=\mathscr D(\Sigma_h)
={}&\frac12\int_{\Sigma_h}
 \bigl(R_{\bar g}-2|E|^2-2\Lambda\bigr)\,dA_h
 +\frac12\int_{\Sigma_h}|\IIo_h|^2\,dA_h\\
&+\int_{\Sigma_h}|E_h^\top|^2\,dA_h
 +\int_{\Sigma_h}
 \left(E_{\nu_h}-\frac{q}{r_h^2}\right)^2dA_h .
\end{align*}
Each integrand is continuous and nonnegative.  Therefore each vanishes
identically, and hence
\[
R_{\bar g}=2|E|^2+2\Lambda,\qquad
\IIo_h=0,\qquad
E_h^\top=0,\qquad
E_{\nu_h}=\frac q{r_h^2}
\]
pointwise on $\Sigma_h$.

If $q=0$, the last two equalities give $E\equiv0$ on $\Omega$, and
therefore $\zeta=0$.  If $q\ne0$, then the same-sign condition holds
and $q_*=q$.  The slab flux identity yields
\[
0=4\pi(q_*-q)=\int_\Omega\zeta\,dV_{\bar g}.
\]
Since $\varsigma\zeta\ge0$ throughout $\Omega$, it follows that
$\zeta\equiv0$.
\end{proof}

For $q\ne0$, vanishing of the defect gives
\[
R_{\bar g}=2|E|^2+2\Lambda,
\qquad
E=\frac{q}{r_h^2}\nu_h
\]
on each leaf, while $\zeta\equiv0$ follows from the constancy of the
leaf charge, the flux law, and the same-sign condition.  For $q=0$,
vanishing of the electromagnetic terms gives $E\equiv0$, and hence
$\zeta\equiv0$.

\subsection{Reconstruction}
\label{subsec:rn-reconstruction}

\begin{proposition}[Reconstruction]
\label{prop:reconstruction}
Under the conclusions of
Proposition~\textup{\ref{prop:vanishing-production}},
$(\Omega,\bar g,E)$ is isometric to the smooth proper-distance closure
of the model annulus $\{r<R\le r_*\}$ with charge $q$, mass $m_0$, and
$\sigma_\kappa=r^{-2}\gamma_0$.
\end{proposition}

\begin{proof}
Since $w_h=0$, the lapse $\rho_h=\bar\rho_h$ is leafwise constant, so
$\mathcal J_h1=\bar\rho_h^{-1}$ is leafwise constant.  With
\eqref{eq:model-equalities}, the Gauss equation gives
\[
 \mathcal J_h1
 =K_h-\frac{q^2}{r_h^4}-\Lambda-\frac34h^2,
\]
so $K_h$ is constant, and Gauss--Bonnet yields $K_h=\kappa/r_h^2$.

After reparametrizing the embeddings by diffeomorphisms of
$\Sigma_0$, we may assume that the variation is normal.  The change of
variable
\[
ds=\bar\rho_h\,dh
\]
then gives a unit normal coordinate, so that
\[
\bar g=ds^2+\gamma_s,
\]
where $\gamma_s$ is the metric induced on the corresponding leaf.

Let $R(s)$ denote its area radius.  Since $\mathring{\II}_s=0$, the metric variation formula gives
$\partial_s\gamma_s=H_s\gamma_s$.  On the other hand, the first
variation of area and $|\Sigma_s|=4\pi R(s)^2$ give
$R'(s)/R(s)=H_s/2$.  Hence $R(s)^{-2}\gamma_s$ is independent of $s$.
Writing
\[
\sigma_\kappa=r^{-2}\gamma_0,
\]
we obtain
\[
\bar g=ds^2+R(s)^2\sigma_\kappa,
\qquad
E=\frac{q}{R(s)^2}\partial_s,
\qquad
H_s=\frac{2R'(s)}{R(s)}.
\]
By the preceding calculation, $\sigma_\kappa$ has area $4\pi$ and
Gaussian curvature $\kappa$.  For $\kappa\le0$, it is the normalized
horizon metric.

Since the fixed-charge mass is identically $m_0$ and
$H_s=2R'(s)/R(s)$, the CMC mass formula gives
\[
R'(s)^2=f_0(R(s)),
\]
with $f_0$ as in \eqref{eq:f0}.
At the horizon,
\[
R(0)=r,\qquad R'(0)=0,\qquad
f_0(r)=0,\qquad
f_0'(r)=\frac{\Theta_{\kappa,\Lambda}}r>0.
\]

For $s>0$, we have $R'(s)=H_sR(s)/2>0$, and hence differentiation of
the preceding identity gives
\[
R''(s)=\frac12 f_0'(R(s)).
\]
This equation extends to $s=0$.  Since the corresponding initial-value
problem is invariant under $s\mapsto-s$, uniqueness shows that $R$ is
even, and therefore
\begin{equation}\label{eq:even-horizon-expansion}
R(s)=r+\frac14f_0'(r)s^2+O(s^4).
\end{equation}
Thus the warped-product metric and the electric field extend smoothly
across $s=0$.

On $s>0$, the function $R=R(s)$ is a coordinate.  Since
$dR=\sqrt{f_0(R)}\,ds$, we obtain
\[
\bar g=\frac{dR^2}{f_0(R)}+R^2\sigma_\kappa,
\qquad
E=\frac q{R^2}\sqrt{f_0(R)}\,\partial_R.
\]
Hence $(\Omega,\bar g,E)$ is the corresponding spatial
Reissner--Nordstr\"om--(anti-)de Sitter collar.
\end{proof}

The preceding propositions establish the implication from the
fixed-charge equality to the model geometry.  We now complete the proof
of Theorem~\ref{thm:main}.

\begin{proof}[Proof of Theorem~\ref{thm:main}]
\textup{(i)}
By \eqref{eq:current-fixed-mass},
\[
m_{\Lambda,q_*}(\Sigma_*)-m_{\Lambda,q}(\Sigma_*)
=\frac{q_*^2-q^2}{2r_*}.
\]
Thus equality holds if and only if $q_*^2=q^2$.  Under
\eqref{eq:same-sign-intro},
\[
\varsigma q_*\ge\varsigma q\ge0,
\]
so this is equivalent to $q_*=q$.  The slab flux law then gives
\[
0=4\pi(q_*-q)=\int_\Omega\zeta\,dV_{\bar g}.
\]
Since $\varsigma\zeta\ge0$, it follows that $\zeta\equiv0$ on
$\Omega$.

Conversely, if $\zeta\equiv0$, then
\eqref{eq:flux-slab} gives $q_h=q$ for every $h$.  In particular,
$q_*=q$, and \eqref{eq:current-fixed-mass} gives
\[
m_{\Lambda,q_*}(\Sigma_*)=m_{\Lambda,q}(\Sigma_*).
\]

\textup{(ii)}
Suppose first that
\[
m_{\Lambda,q}(\Sigma_*)=m_{\Lambda,q}(\Sigma_0).
\]
Proposition~\ref{prop:vanishing-production} shows that the defect and
the lapse fluctuation vanish on every leaf, that $q_h=q$, and that
$\zeta\equiv0$ throughout the collar.  Proposition
~\ref{prop:reconstruction} then identifies $(\Omega,\bar g,E)$ with
the stated spatial Reissner--Nordstr\"om--(anti-)de Sitter collar.

Conversely, suppose that $(\Omega,\bar g,E)$ is the stated model
collar.  On each coordinate surface $\Sigma_R$, the electric flux is
$q$ and
\[
H=\frac{2\sqrt{f_0(R)}}{R}.
\]
Therefore \eqref{eq:hawking-cmc} gives
\[
\begin{aligned}
m_{\Lambda,q}(\Sigma_R)
=\frac{R}{2}
 \left(
 \kappa+\frac{q^2}{R^2}
 -\frac{\Lambda R^2}{3}
 -f_0(R)
 \right)
=m_0.
\end{aligned}
\]
Hence
\[
m_{\Lambda,q}(\Sigma_*)
=m_{\Lambda,q}(\Sigma_0).
\]
\end{proof}

The positive-Jacobi hypothesis restricts the radial extent of the
reconstructed model collar.

\begin{remark}[Positive-Jacobi range of the model]
\label{rem:model-range}
On a model coordinate leaf of radius $R$, the first Jacobi eigenvalue is
\[
\lambda_1(\mathcal J_R^{\mathrm{mod}})
=
-\frac{2}{R^4}
\bigl(\kappa R^2-3m_0R+2q^2\bigr).
\]
Thus the positive-Jacobi condition determines the admissible radial
range of the model collar.

If $\kappa=1$, positivity holds precisely for
\[
r\le R<R_{\mathrm{ph}}^+(m_0,q),
\qquad
R_{\mathrm{ph}}^+(m_0,q)
=
\frac12\left(3m_0+\sqrt{9m_0^2-8q^2}\right),
\]
and hence $r_*<R_{\mathrm{ph}}^+(m_0,q)$.  If $\kappa\le0$, the first
Jacobi eigenvalue remains positive for every $R\ge r$, so the
positive-Jacobi condition imposes no finite upper radial cutoff.
\end{remark}

\subsection{Local rigidity from local maximality}
\label{subsec:local-rigidity}

The rigidity theorem above yields a local result when the horizon is a
local maximizer of the fixed-charge Hawking mass.

\begin{corollary}[Local rigidity from local maximality]
\label{cor:local-maximality}
Let $\Sigma$ be a connected strictly stable minimal surface, fix a unit
normal $\nu$, and set
\[
q:=\frac1{4\pi}\int_\Sigma E_\nu\,dA.
\]
\begin{enumerate}[label=\textup{(\roman*)},itemsep=2pt]
\item Suppose that, on the $\nu$-side of $\Sigma$,
\[
R_{\bar g}\ge2|E|^2+2\Lambda,
\]
and that either $q=0$ or \eqref{eq:same-sign-intro} holds.  If $\Sigma$
locally maximizes $m_{\Lambda,q}$ among nearby normal graphs on that
side, then a sufficiently small one-sided neighborhood of $\Sigma$ is
isometric to the corresponding spatial
Reissner--Nordstr\"om--(anti-)de Sitter model collar.

\item Suppose that
\[
R_{\bar g}\ge2|E|^2+2\Lambda
\]
in a full neighborhood of $\Sigma$, and that either $q=0$ or
\eqref{eq:same-sign-intro} holds on each side of $\Sigma$ with respect
to the unit normal pointing into that side.  If $\Sigma$ locally
maximizes $m_{\Lambda,q}$ among all nearby normal graphs, then a
two-sided neighborhood of $\Sigma$ is isometric to the corresponding
two-sided model neighborhood; in particular, $\zeta\equiv0$ on that
neighborhood.
\end{enumerate}
\end{corollary}

\begin{proof}
\textup{(i)}
The short canonical CMC branch is a charged positive-Jacobi CMC region.
Monotonicity and local maximality give, respectively,
\[
m_{\Lambda,q}(\Sigma_h)\ge m_{\Lambda,q}(\Sigma)
\qquad\text{and}\qquad
m_{\Lambda,q}(\Sigma_h)\le m_{\Lambda,q}(\Sigma).
\]
Thus the mass is constant along the branch, and
Theorem~\ref{thm:main}\textup{(ii)} gives the model collar.

\textup{(ii)}
Apply part~\textup{(i)} on each side of $\Sigma$, for the unit normal
pointing into that side; by Proposition~\ref{prop:vanishing-production}
both collars are source-free.  Since
$m_{\Lambda,-q}=m_{\Lambda,q}$, the two reconstructed collars have the
same horizon data.  Their warping functions are therefore restrictions
of the same even solution, and the metric and electric field glue
smoothly across $\Sigma$.
\end{proof}

In the corresponding uncharged settings, this corollary recovers the
local Hawking-mass rigidity results of M\'aximo--Nunes
\cite{MaximoNunes} and Barros--Batista--Cruz
\cite{BarrosBatistaCruz}.  Under conserved electric flux, it also gives
the charged counterparts obtained by Baltazar--Barros--Batista
\cite{BaltazarBarrosBatista} and Sousa--Lima \cite{SousaLima}.  The
formulation above treats all signs of $\Lambda$ in a single argument.
Moreover, the one-sided statement permits charged matter satisfying
the same-sign condition; local maximality then forces the source to
vanish on the resulting model collar.

For $q\ne0$, reversing the normal changes the sign of the oriented
charge but not that of $\zeta=\operatorname{div}_{\bar g}E$.  The
hypothesis of part~\textup{(ii)} thus reads $q\zeta\ge0$ on the
$\nu$-side and $q\zeta\le0$ on the $-\nu$-side.  It permits sources
that change sign across $\Sigma$: Appendix~\ref{app:sourced-equality-collars}
with $Q(s)=q+as^2$, $q,a>0$, continued to $s<0$ by evenness, has
$\zeta=2as/R(s)^2$.

\section{Horizon Hessian and endpoint rigidity}
\label{sec:endpoint-hessian}

Let $\Sigma_0=\Sigma$ be the strictly stable minimal endpoint of the
canonical CMC branch, with
\[
 |\Sigma|=4\pi r^2,
 \qquad
 q=q_0.
\]
Throughout this section, $m_{\Lambda,q}$ denotes the Hawking mass with
the charge parameter frozen at $q$.  We set
\[
 \Theta
 :=
 \Theta_{\kappa,\Lambda}
 =
 \kappa-\frac{q^2}{r^2}-\Lambda r^2>0,
\]
where the positivity follows from
Lemma~\ref{lem:constant-slack}.

\subsection{The fixed-charge Hessian}
\label{subsec:fixed-charge-hessian}

\begin{proposition}[Second variation at the horizon]
\label{prop:second-variation}
The surface $\Sigma$ is a critical point of the fixed-charge Hawking
mass $m_{\Lambda,q}$.  For every smooth normal variation
$\{\Sigma_t\}$ with $\Sigma_0=\Sigma$ and initial normal speed
$\varphi$, its second variation is
\begin{equation}\label{eq:second-variation}
 \left.m_{\Lambda,q}''\right|_\varphi
 =
 \frac r{16\pi}
 \left[
  \frac{\Theta}{r^2}\mathcal Q(\varphi,\varphi)
  -
  \int_\Sigma(\mathcal J\varphi)^2\,dA
 \right].
\end{equation}
\end{proposition}

\begin{proof}
Set
\[
 A(t):=|\Sigma_t|,
 \qquad
 I(t):=\int_{\Sigma_t}H_t^2\,dA_t .
\]
Since $H_0=0$,
\[
 A'(0)=I'(0)=0,
\]
and hence $\Sigma$ is a critical point of $m_{\Lambda,q}$.  The standard
variation formulas give
\[
 A''(0)=\mathcal Q(\varphi,\varphi),
 \qquad
 I''(0)=2\int_\Sigma(\mathcal J\varphi)^2\,dA .
\]

Along the variation, the topology and the charge parameter are fixed,
so
\[
 m_{\Lambda,q}
 =
 \sqrt{\frac{A}{16\pi}}
 \left(
  \kappa+\frac{4\pi q^2}{A}
  -\frac{\Lambda A}{12\pi}
  -\frac{I}{16\pi}
 \right).
\]
At $A=4\pi r^2$ and $I=0$,
\[
 \frac{\partial m_{\Lambda,q}}{\partial A}
 =
 \frac{\Theta}{16\pi r},
 \qquad
 \frac{\partial m_{\Lambda,q}}{\partial I}
 =
 -\frac r{32\pi}.
\]
Since $A'(0)=I'(0)=0$, the chain rule therefore yields
\[
 \left.m_{\Lambda,q}''\right|_\varphi
 =
 \frac{\Theta}{16\pi r}\mathcal Q(\varphi,\varphi)
 -
 \frac r{16\pi}
 \int_\Sigma(\mathcal J\varphi)^2\,dA,
\]
which is \eqref{eq:second-variation}.
\end{proof}

\subsection{The canonical direction and endpoint rigidity}
\label{subsec:endpoint-rigidity}

Let $\rho>0$ be the unique solution of
\[
 \mathcal J\rho=1.
\]
Set
\[
 \bar\rho:=\frac1{|\Sigma|}\int_\Sigma\rho\,dA,
 \qquad
 w:=\rho-\bar\rho,
 \qquad
 p:=\frac{\rho}{\sqrt{\bar\rho}}.
\]
Then
\[
 \mathcal Jp=\bar\rho^{-1/2},
 \qquad
 \mathcal Q(p,p)=|\Sigma|.
\]

\begin{proposition}[Canonical endpoint identity]
\label{prop:endpoint-identity}
The fixed-charge Hessian in the direction $p$ satisfies
\begin{equation}\label{eq:endpoint-identity}
 \frac{16\pi}{r}
 \left.m_{\Lambda,q}''\right|_p
 =
 4\pi\Theta-\frac{|\Sigma|}{\bar\rho}
 =
 \mathscr D(\Sigma)
 +\frac1{\bar\rho^2}\mathcal Q(w,w)
 =
 \mathscr B_0
 \ge0.
\end{equation}
Moreover, along the canonical CMC branch,
\begin{equation}\label{eq:radial-endpoint}
 \left.
 \frac{d^2}{dh^2}
 \right|_{h=0}
 m_{\Lambda,q}(\Sigma_h)
 =
 \bar\rho
 \left.m_{\Lambda,q}''\right|_p .
\end{equation}
\end{proposition}

\begin{proof}
Since $\mathcal Q(p,p)=|\Sigma|$ and
$\mathcal Jp=\bar\rho^{-1/2}$,
Proposition~\ref{prop:second-variation} gives
\[
 \frac{16\pi}{r}
 \left.m_{\Lambda,q}''\right|_p
 =
 4\pi\Theta-\frac{|\Sigma|}{\bar\rho}.
\]
By self-adjointness,
\[
 |\Sigma|
 =
 \mathcal Q(1,\rho)
 =
 \bar\rho\,\mathcal Q(1,1)+\mathcal Q(1,w),
\]
while
\[
 0
 =
 \mathcal Q(\rho,w)
 =
 \bar\rho\,\mathcal Q(1,w)+\mathcal Q(w,w).
\]
Hence
\[
 \frac{|\Sigma|}{\bar\rho}
 =
 \mathcal Q(1,1)
 -
 \frac1{\bar\rho^2}\mathcal Q(w,w).
\]
Combining this with
\[
 4\pi\Theta-\mathcal Q(1,1)=\mathscr D(\Sigma)
\]
from Lemma~\ref{lem:constant-slack} proves
\eqref{eq:endpoint-identity}.

The initial normal speed of the canonical CMC branch is $\rho$.
Since $\Sigma$ is critical for $m_{\Lambda,q}$ and
$\rho=\sqrt{\bar\rho}\,p$, we obtain
\[
 \left.
 \frac{d^2}{dh^2}
 \right|_{h=0}
 m_{\Lambda,q}(\Sigma_h)
 =
 \left.m_{\Lambda,q}''\right|_\rho
 =
 \bar\rho
 \left.m_{\Lambda,q}''\right|_p .
\]
\end{proof}

\begin{corollary}[Endpoint rigidity]
\label{cor:endpoint-rigidity}
The following are equivalent:
\begin{enumerate}[label=\textup{(\roman*)},itemsep=2pt]
\item
\[
 \left.m_{\Lambda,q}''\right|_p=0;
\]

\item
\[
 \left.m_{\Lambda,q}''\right|_\varphi\le0
 \qquad
 \text{for every }\varphi\in C^\infty(\Sigma);
\]

\item
\[
 \lambda_1(\mathcal J)r^2=\Theta;
\]

\item $\Sigma$ carries the model horizon data
\begin{equation}\label{eq:model-data}
 K=\frac{\kappa}{r^2},
 \qquad
 \II=0,
 \qquad
 E=\frac q{r^2}\nu,
 \qquad
 R_{\bar g}=2|E|^2+2\Lambda
 \quad\text{on }\Sigma.
\end{equation}
\end{enumerate}
\end{corollary}

\begin{proof}
The implication \textup{(iv)}$\Rightarrow$\textup{(i)} is immediate
from \eqref{eq:endpoint-identity}.

Assume \textup{(i)}.  Then
\[
 \mathscr D(\Sigma)=0,
 \qquad
 w=0.
\]
Hence $\rho$ is constant and
\[
 \mathcal J=-\Delta_\Sigma+c,
 \qquad
 c=\frac{\mathcal Q(1,1)}{|\Sigma|}
   =\frac{\Theta}{r^2}.
\]
Therefore \eqref{eq:second-variation} gives
\[
 \frac{16\pi}{r}
 \left.m_{\Lambda,q}''\right|_\varphi
 =
 -\int_\Sigma
 \left((\Delta_\Sigma\varphi)^2
 +\frac{\Theta}{r^2}|\nabla\varphi|^2\right)dA
 \le0,
\]
proving \textup{(i)}$\Rightarrow$\textup{(ii)}.

If \textup{(ii)} holds, applying \eqref{eq:second-variation} to a
normalized first eigenfunction gives
\[
 0\ge
 \lambda_1\left(\frac{\Theta}{r^2}-\lambda_1\right).
\]
Thus $\lambda_1r^2\ge\Theta$, while
Lemma~\ref{lem:constant-slack} gives the reverse inequality.  Hence
\textup{(iii)} holds.

Finally, under \textup{(iii)},
\[
 \lambda_1|\Sigma|
 \le\mathcal Q(1,1)
 =4\pi\Theta-\mathscr D(\Sigma)
 \le4\pi\Theta
 =\lambda_1|\Sigma|.
\]
Thus $\mathscr D(\Sigma)=0$ and equality holds in the Rayleigh
quotient for the constant function, so
$\mathcal J1=\lambda_1$.  The vanishing of the defect gives
\[
 R_{\bar g}=2|E|^2+2\Lambda,
 \qquad
 \II=0,
 \qquad
 E=\frac q{r^2}\nu,
\]
and the Gauss equation yields $K=\kappa/r^2$.  Hence \textup{(iv)}
holds.
\end{proof}

Consequently, if $\Sigma$ does not carry the model horizon data, then
\[
 \left.m_{\Lambda,q}''\right|_p>0,
 \qquad
 \left.
 \frac{d^2}{dh^2}
 \right|_{h=0}
 m_{\Lambda,q}(\Sigma_h)
 =
 \bar\rho
 \left.m_{\Lambda,q}''\right|_p
 >0.
\]
Thus the canonical CMC branch rules out local maximality of the
fixed-charge Hawking mass at every nonmodel horizon.
The implications
\textup{(ii)}$\Rightarrow$\textup{(iii)}$\Rightarrow$\textup{(iv)}
follow \cite{MaximoNunes} and \cite[Prop.~4]{BaltazarBarrosBatista};
the equivalence with \textup{(i)} rests on \eqref{eq:endpoint-identity}.

These equivalent conditions characterize only the horizon data along
$\Sigma$; they do not determine the surrounding collar.
Appendix~\ref{app:concentrated-shell} constructs nonmodel collars whose
horizons carry the exact model data.  Rigidity of the full collar
instead follows from the fixed-charge boundary equality in
Theorem~\ref{thm:main}.

\begin{remark}[Why the charge is frozen]
Let $q_t$ be the flux through a normal variation with initial speed
$\varphi$.  Then
\[
 \left.\frac d{dt}\right|_{t=0}
 m_{\Lambda,q_t}(\Sigma_t)
 =
 \frac q{4\pi r}\int_\Sigma\varphi\zeta\,dA.
\]
Hence the current-charge mass is critical for all normal variations
only when $q\zeta\equiv0$ on $\Sigma$; compare
\cite[Remark~2.20]{PenuelaDiaz}.  The fixed-charge mass is
critical at every minimal endpoint, which is why the Hessian above is
formulated with the charge parameter frozen.
\end{remark}

\appendix
\section{Sourced CMC collars and a concentrated shell}
\label{app:matter-example}

\subsection{Sourced collars with scalar-constraint equality}
\label{app:sourced-equality-collars}

Fix a closed connected orientable surface $\Sigma$ carrying a constant-curvature
metric $\sigma_\kappa$ of curvature $\kappa$ and area $4\pi$, so that
$\kappa=\chi/2$.  Let $r>0$ and $q\in\mathbb R$ satisfy
\[
 \Theta
 :=
 \kappa-\frac{q^2}{r^2}-\Lambda r^2
 >0.
\]

\begin{proposition}[Sourced collars with scalar-constraint equality]
\label{prop:sourced-equality-collars}
Let $Q$ be smooth near $s=0$ with $Q(0)=q$.  Then, for some
$s_1>0$, the initial-value problem
\begin{equation}\label{eq:sourced-warping-ode}
 2RR''+R'^2
 =
 \kappa-\frac{Q(s)^2}{R^2}-\Lambda R^2,
 \qquad
 R(0)=r,
 \qquad
 R'(0)=0,
\end{equation}
has a smooth solution satisfying
\[
 R'(s)>0\quad (s>0),
 \qquad
 \frac d{ds}\left(\frac{2R'}R\right)>0
 \quad\text{on }[0,s_1].
\]
The data
\begin{equation}\label{eq:sourced-proper-distance-data}
 \bar g=ds^2+R(s)^2\sigma_\kappa,
 \qquad
 E=\frac{Q(s)}{R(s)^2}\,\partial_s
\end{equation}
on $\Omega=[0,s_1]\times\Sigma$ form a charged positive-Jacobi CMC
region with minimal horizon $\Sigma_0$.  Moreover,
\begin{equation}\label{eq:sourced-collar-identities}
 R_{\bar g}=2|E|^2+2\Lambda,
 \qquad
 \operatorname{div}_{\bar g}E=\frac{Q'}{R^2}.
\end{equation}
In particular, $Q'\not\equiv0$ produces a genuine charge source.
\end{proposition}

\begin{proof}
Standard ODE theory gives a unique smooth local solution of
\eqref{eq:sourced-warping-ode}.  At $s=0$,
\[
 R''(0)=\frac{\Theta}{2r}>0,
 \qquad
 \left.\frac d{ds}\left(\frac{2R'}R\right)\right|_{s=0}
 =\frac{\Theta}{r^2}>0.
\]
After reducing $s_1$, the stated inequalities follow.

The leaves $\Sigma_s=\{s\}\times\Sigma$ are umbilic, with
\[
 H(s)=\frac{2R'}R.
\]
Since the Jacobi potential is leafwise constant,
\[
 \mathcal J_s
 =
 -\Delta_{\Sigma_s}+H'(s),
\]
and hence $H'>0$ gives the positive-Jacobi condition.  Parametrizing
the leaves by $h=H(s)$, the lapse is $\rho=1/H'$, which is positive
and leafwise constant.

Finally,
\[
 R_{\bar g}
 =
 -\frac{4R''}R
 +\frac{2(\kappa-R'^2)}{R^2},
 \qquad
 |E|^2=\frac{Q^2}{R^4},
 \qquad
 \operatorname{div}_{\bar g}E=\frac{Q'}{R^2}.
\]
The first identity in \eqref{eq:sourced-collar-identities} is therefore
equivalent to \eqref{eq:sourced-warping-ode}.
\end{proof}

At the horizon,
\[
 K=\frac{\kappa}{r^2},
 \qquad
 \II=0,
 \qquad
 E=\frac q{r^2}\nu,
 \qquad
 R_{\bar g}=2|E|^2+2\Lambda,
\]
so $\Sigma_0$ carries the model horizon data even when a charge source
is present immediately outside it.

When $Q\equiv q$, equation \eqref{eq:sourced-warping-ode} has the first
integral
\begin{equation}\label{eq:sourced-model-first-integral}
 R'^2
 =
 \kappa-\frac{2m_0}{R}
 +\frac{q^2}{R^2}
 -\frac{\Lambda R^2}{3},
 \qquad
 m_0
 =
 \frac r2
 \left(
  \kappa+\frac{q^2}{r^2}
  -\frac{\Lambda r^2}{3}
 \right),
\end{equation}
and the construction reduces to the source-free
Reissner--Nordstr\"om--(anti-)de Sitter collar.

If
\[
 \varsigma q\ge0,
 \qquad
 \varsigma Q'\ge0
\]
for some $\varsigma\in\{-1,1\}$, then the same-sign condition
\eqref{eq:same-sign-intro} holds.  On the other hand, the cosmological
analogue of McCormick's pointwise condition
\cite[(3.16)]{McCormickMatter},
\begin{equation}\label{eq:mccormick-pointwise}
 R_{\bar g}-2|E|^2-2\Lambda
 \ge4\left|\operatorname{div}_{\bar g}E\right|,
\end{equation}
would force $Q'\equiv0$ in this equality case.  Thus
\eqref{eq:mccormick-pointwise} excludes the genuinely sourced collars
above, whereas the same-sign condition permits them.
\subsection{A concentrated scalar-constraint shell}
\label{app:concentrated-shell}

The endpoint result characterizes the data on the horizon, but does
not determine the surrounding collar.  We now show that even exact
model horizon data together with an arbitrarily small fixed-charge
mass gap do not give $C^2$ control of the collar.

Let
\[
 \bar g_0=ds^2+R_0(s)^2\sigma_\kappa,
 \qquad
 E_0=\frac q{R_0(s)^2}\partial_s
\]
be a source-free model positive-Jacobi CMC collar on
$\Omega=[0,s_*]\times\Sigma$, with
\[
 R_0(0)=r,
 \qquad
 R_0'(0)=0.
\]
Write
\[
 H_0=\frac{2R_0'}{R_0},
 \qquad
 \mu:=\min_{[0,s_*]}H_0'>0,
\]
and let $m_0$ be its constant fixed-charge Hawking mass.

We keep the electric flux equal to $q$ and insert a nonnegative
scalar-constraint defect into a shrinking shell.  For radial data
\[
 \bar g=ds^2+R(s)^2\sigma_\kappa,
 \qquad
 E=\frac q{R(s)^2}\partial_s,
\]
one has
\[
 \operatorname{div}_{\bar g}E=0
\]
and
\begin{equation}\label{eq:shell-scalar-formula}
 R_{\bar g}-2|E|^2-2\Lambda
 =
 \frac2{R^2}
 \left(
  \kappa-R'^2-2RR''
  -\frac{q^2}{R^2}-\Lambda R^2
 \right).
\end{equation}
Thus prescribing
\begin{equation}\label{eq:shell-prescribed-defect}
 R_{\bar g}-2|E|^2-2\Lambda
 =
 \frac{4\eta(s)}{R^2}
\end{equation}
amounts to solving
\begin{equation}\label{eq:shell-ode}
 2RR''+R'^2
 =
 \kappa-\frac{q^2}{R^2}
 -\Lambda R^2-2\eta.
\end{equation}
For the fixed-charge Hawking mass of $\Sigma_s=\{s\}\times\Sigma$,
this equation gives
\begin{equation}\label{eq:shell-mass-derivative}
 \frac d{ds}m_{\Lambda,q}(\Sigma_s)
 =
 R'(s)\eta(s).
\end{equation}
Thus the scalar-constraint defect is pointwise proportional to
$\eta$, while the fixed-charge mass gap depends on the weighted
integral of $\eta$.

Choose
\[
 0\le\beta\in C_c^\infty((1,2)),
 \qquad
 \max\beta=1,
\]
fix a sufficiently small $a>0$, and set
\[
 \eta_\varepsilon(s)
 :=
 a\beta\left(\frac{s}{\varepsilon}\right).
\]
Let $R_\varepsilon$ solve \eqref{eq:shell-ode} with
$\eta=\eta_\varepsilon$ and
\[
 R_\varepsilon(0)=r,
 \qquad
 R_\varepsilon'(0)=0.
\]
Since $0\le\eta_\varepsilon\le a$ is supported in
$[\varepsilon,2\varepsilon]$, Gronwall's inequality gives, for all
sufficiently small $\varepsilon$, existence on $[0,s_*]$ and
\begin{equation}\label{eq:shell-c1-estimate}
 \|R_\varepsilon-R_0\|_{C^1([0,s_*])}
 \le C\,\|\eta_\varepsilon\|_{L^1(0,s_*)}
 =O(\varepsilon),
\end{equation}
with $C$ independent of $\varepsilon$.  Moreover,
\[
 \left(\frac{2R_\varepsilon'}{R_\varepsilon}\right)'
 =
 \frac{
  \kappa-q^2/R_\varepsilon^2
  -\Lambda R_\varepsilon^2
  -3R_\varepsilon'^2
  -2\eta_\varepsilon
 }{R_\varepsilon^2}.
\]
Choosing $a$ sufficiently small preserves
\[
 \left(\frac{2R_\varepsilon'}{R_\varepsilon}\right)'
 \ge\frac{\mu}{2},
\]
so
\[
 \bar g_\varepsilon
 =
 ds^2+R_\varepsilon(s)^2\sigma_\kappa,
 \qquad
 E_\varepsilon
 =
 \frac q{R_\varepsilon(s)^2}\partial_s
\]
is again a source-free charged positive-Jacobi CMC collar.

Since $\eta_\varepsilon=0$ on $[0,\varepsilon]$, uniqueness gives
\[
 R_\varepsilon=R_0
 \qquad\text{on }[0,\varepsilon].
\]
Thus $(\bar g_\varepsilon,E_\varepsilon)$ agrees exactly with the model
near the horizon.  In particular, the horizon carries the model data
and
\[
 \left.m_{\Lambda,q}''\right|_p=0.
\]

On the other hand, \eqref{eq:shell-mass-derivative} gives
\[
 \begin{aligned}
 m_{\Lambda,q}(\Sigma_{s_*};\bar g_\varepsilon)-m_0
 &=
 \int_\varepsilon^{2\varepsilon}
 R_\varepsilon'(s)\eta_\varepsilon(s)\,ds  \\
 &>0.
 \end{aligned}
\]
Since $R_\varepsilon'(s)=O(\varepsilon)$ on the shell,
\begin{equation}\label{eq:shell-small-mass-gap}
 m_{\Lambda,q}(\Sigma_{s_*};\bar g_\varepsilon)-m_0
 =
 O(\varepsilon^2).
\end{equation}

Finally, choose $t_0\in(1,2)$ with $\beta(t_0)=1$ and put
$s_\varepsilon=\varepsilon t_0$.  Then
\[
 R_\varepsilon(s_\varepsilon)=r+O(\varepsilon^2),
 \qquad
 R_\varepsilon(s_\varepsilon)-R_0(s_\varepsilon)
 =
 O(\varepsilon^2).
\]
Using \eqref{eq:shell-prescribed-defect} and the model equality,
\[
 \begin{aligned}
 R_{\bar g_\varepsilon}(s_\varepsilon)
 -
 R_{\bar g_0}(s_\varepsilon)
 ={}&
 2q^2
 \left(
  \frac1{R_\varepsilon(s_\varepsilon)^4}
  -
  \frac1{R_0(s_\varepsilon)^4}
 \right)
 +
 \frac{4a}{R_\varepsilon(s_\varepsilon)^2}  \\
 ={}&
 \frac{4a}{r^2}+O(\varepsilon^2).
 \end{aligned}
\]
Therefore $\bar g_\varepsilon\to\bar g_0$ in $C^1$ by
\eqref{eq:shell-c1-estimate}, but not in $C^2$, even though the horizon
is exactly model and the fixed-charge mass gap tends to zero.

\section*{Acknowledgements}
\begin{sloppypar}
The author thanks Hyun Chul Jang for helpful discussions on this work.

This research was supported by Global - Learning \& Academic research
institution for Master's\textperiodcentered PhD students, and Postdocs
(G-LAMP) Program of the National Research Foundation of Korea (NRF)
grant funded by the Ministry of Education (No.~RS-2025-25442252).
\end{sloppypar}

During the preparation of this manuscript the author used large
language models to check the
computations against the written text, to verify bibliographic data,
and to improve the readability of the exposition.  All mathematical
content was conceived and verified by the author.


\begin{thebibliography}{99}

\bibitem{AlaeeCabreraCederbaum}
A.~Alaee, A.~J.~Cabrera Pacheco, and C.~Cederbaum,
\emph{Asymptotically flat extensions with charge},
Adv. Theor. Math. Phys. \textbf{23} (2019), 1951--1980.

\bibitem{BaltazarBarrosBatista}
H.~Baltazar, A.~Barros, and R.~Batista,
\emph{A local rigidity theorem for minimal two-spheres in charged
time-symmetric initial data set},
Lett. Math. Phys. \textbf{113} (2023), 92.

\bibitem{BarrosBatistaCruz}
A.~Barros, R.~Batista, and T.~Cruz,
\emph{Hawking mass and local rigidity of minimal surfaces in
three-manifolds},
Comm. Anal. Geom. \textbf{25} (2017), 1--23.

\bibitem{CabreraCederbaumGehringPenuela}
A.~J.~Cabrera Pacheco, C.~Cederbaum, P.~Gehring, and A.~Pe\~nuela Diaz,
\emph{Constructing electrically charged Riemannian manifolds with minimal
boundary, prescribed asymptotics, and controlled mass},
J. Geom. Phys. \textbf{185} (2023), 104746.

\bibitem{ChenMcCormick}
P.-N.~Chen and S.~McCormick,
\emph{Quasi-local Penrose inequalities with electric charge},
Int. Math. Res. Not. IMRN (2022), no.~22, 17333--17362.

\bibitem{DainJaramilloReiris}
S.~Dain, J.~L.~Jaramillo, and M.~Reiris,
\emph{Area--charge inequality for black holes},
Class. Quantum Grav. \textbf{29} (2012), 035013.

\bibitem{DisconziKhuri}
M.~M.~Disconzi and M.~A.~Khuri,
\emph{On the Penrose inequality for charged black holes},
Class. Quantum Grav. \textbf{29} (2012), 245019.

\bibitem{GallowayMendesCharged}
G.~J.~Galloway and A.~Mendes,
\emph{Some rigidity results for charged initial data sets},
Nonlinear Anal. \textbf{256} (2025), 113780.

\bibitem{HaywardShiromizuNakao}
S.~A.~Hayward, T.~Shiromizu, and K.-i.~Nakao,
\emph{A cosmological constant limits the size of black holes},
Phys. Rev. D \textbf{49} (1994), 5080--5085.

\bibitem{Jang}
P.~S.~Jang,
\emph{Note on cosmic censorship},
Phys. Rev. D \textbf{20} (1979), 834--838.

\bibitem{KWY}
M.~Khuri, G.~Weinstein, and S.~Yamada,
\emph{Proof of the Riemannian Penrose inequality with charge for
multiple black holes},
J. Differential Geom. \textbf{106} (2017), 451--498.

\bibitem{KWYExtensions}
M.~Khuri, G.~Weinstein, and S.~Yamada,
\emph{Extensions of the charged Riemannian Penrose inequality},
Class. Quantum Grav. \textbf{32} (2015), 035019.

\bibitem{LimaSousaBatista}
A.~B.~Lima, P.~A.~Sousa, and R.~M.~Batista,
\emph{Rigidity of marginally outer trapped surfaces in charged initial
data sets},
Lett. Math. Phys. \textbf{115} (2025), 41.

\bibitem{MaximoNunes}
D.~M\'aximo and I.~Nunes,
\emph{Hawking mass and local rigidity of minimal two-spheres in
three-manifolds},
Comm. Anal. Geom. \textbf{21} (2013), 409--433.

\bibitem{McCormickFillins}
S.~McCormick,
\emph{Fill-ins with scalar curvature lower bounds and applications to
positive mass theorems},
Ann. Global Anal. Geom. \textbf{65} (2024), 26.

\bibitem{McCormickMatter}
S.~McCormick,
\emph{On the charged Riemannian Penrose inequality with charged matter},
Class. Quantum Grav. \textbf{37} (2020), 015007.

\bibitem{MendesAreaCharge}
A.~Mendes,
\emph{Area-charge inequality and local rigidity in charged initial data
sets},
Class. Quantum Grav. \textbf{42} (2025), 225020.

\bibitem{Miao2009}
P.~Miao,
\emph{On a localized Riemannian Penrose inequality},
Comm. Math. Phys. \textbf{292} (2009), 271--284.

\bibitem{MiaoWangXie}
P.~Miao, Y.~Wang, and N.~Xie,
\emph{On Hawking mass and Bartnik mass of CMC surfaces},
Math. Res. Lett. \textbf{27} (2020), 855--885.

\bibitem{PenuelaDiaz}
A.~Pe\~nuela Diaz,
\emph{Rigidity and positivity of Hawking quasi-local energy on
area-constrained critical surfaces},
preprint (2025), arXiv:2507.16588.

\bibitem{SimonAreaCharge}
W.~Simon,
\emph{Bounds on area and charge for marginally trapped surfaces with
cosmological constant},
Class. Quantum Grav. \textbf{29} (2012), 062001.

\bibitem{SousaLima}
P.~A.~Sousa and A.~B.~Lima,
\emph{Charged Hawking mass and local rigidity of three-manifolds},
J. Geom. Anal. \textbf{33} (2023), 11.

\bibitem{Woolgar}
E.~Woolgar,
\emph{Bounded area theorems for higher genus black holes},
Class. Quantum Grav. \textbf{16} (1999), 3005--3012.

\end{thebibliography}
\end{document}